\documentclass[11pt, a4paper, onecolumn]{article}
\usepackage{textcomp}

\usepackage{amsmath}

\usepackage{subcaption}
\usepackage[left=1.6cm, right=1.6cm, top=1.6cm, bottom=2cm]{geometry}
 
\usepackage{graphicx}
\usepackage[utf8]{inputenc} 
\usepackage[T1]{fontenc}    
\usepackage{hyperref}       
\usepackage{url}            
\usepackage{booktabs}       
\usepackage{amsfonts}       
\usepackage{nicefrac}       
\usepackage{microtype}      
\usepackage{xcolor}         
\usepackage{algorithm}

\newtheorem{theorem}{Theorem}

\newtheorem{lemma}{Lemma}
\newtheorem{remark}{Remark}
\newtheorem{assumption}{Assumption}

\DeclareFontFamily{OT1}{pzc}{}
\DeclareFontShape{OT1}{pzc}{m}{it}{<-> s * [1.200] pzcmi7t}{}
\DeclareMathAlphabet{\mathpzc}{OT1}{pzc}{m}{it}

\newcommand*\mcapinn[2]{\vcenter{\hbox{$\mathsurround=0pt
  \ifx\displaystyle#1\textstyle\else#1\fi\bigcap$}}}

\newcommand*\mcupinn[2]{\vcenter{\hbox{$\mathsurround=0pt
  \ifx\displaystyle#1\textstyle\else#1\fi\bigcup$}}}

\def\begequarr{\begin{eqnarray}}
\def\endequarr{\end{eqnarray}}
\def\begequarrs{\begin{eqnarray*}}
\def\endequarrs{\end{eqnarray*}}
\def\begequ{\begin{equation}}
\def\endequ{\end{equation}}
\def\begequs{\begin{equation*}}
\def\endequs{\end{equation*}}
\def\begite{\begin{itemize}}
\def\endite{\end{itemize}}

\def\begcen{\begin{center}}
\def\endcen{\end{center}}
\def\begrem{\begin{remark}\rm}
\def\endrem{\end{remark}}
\def\ba{\begin{aligned}}
\def\ea{\end{aligned}}

\def\col{\textnormal{col}\; }

\def\col{\textnormal{col}}

\newcommand{\mV}{\mathbb{V}}

\newcommand{\bb}{\mathbf{b}}

\newcommand{\Gb}{\mathbf{G}}

\newcommand{\xb}{\mathbf{x}}

\newcommand{\yb}{\mathbf{y}}

\newcommand{\Hb}{\mathbf{H}}

\newcommand{\Lb}{\mathbf{L}}
\newcommand{\Lc}{\mathcal{L}}

\newcommand{\Xb}{\mathbf{X}}

\newcommand{\Vb}{\mathbf{V}}

\newcommand{\Xt}{\tilde{\mathbf{X}}}

\newcommand{\xt}{\tilde{\mathbf{x}}}
\newcommand{\yt}{\tilde{\mathbf{y}}}

\newcommand{\zt}{\tilde{\mathbf{z}}}
\newcommand{\ft}{\tilde{f}}

\title{\bf Distributed Nesterov Flows for Multi-agent Optimization (Extended Version)} 

\author{%
    Zihao Ren,
    Lei Wang, and
    Guodong Shi\thanks{Z. Ren and L. Wang are with College of Control Science and Engineering, Zhejiang University, Hangzhou, China (Email: zhren200@zju.edu.cn, lei.wangzju@zju.edu.cn). G. Shi is with the Australian Center for Robotics, School of Aerospace, Mechanical, and Mechatronics Engineering, The University of Sydney, Australia (Email: guodong.shi@sydney.edu.au).}
}

\date{}
\begin{document}

\maketitle

\begin{abstract}
Various distributed gradient descent algorithms for multi-agent optimization have incorporated the Nesterov accelerated gradient  method, where the use of momentum enhances convergence rates. These algorithms have found broad applications in large-scale machine learning and optimization owing to their simplicity and low communication complexity. In this paper, we establish a continuous-time approximation of distributed Nesterov gradient descent. The convergence properties and convergence rate of the resulting distributed Nesterov flow are analyzed using Lyapunov methods. Building on these insights, we design new parameter choices within the flow, from which we derive flow-inspired discrete-time algorithms for multi-agent optimization. Surprisingly, the resulting algorithms achieve faster convergence compared to existing distributed gradient descent methods: they require fewer iterations to reach the same accuracy for strongly convex functions and exhibit an improved convergence rate for general convex functions without incurring additional communication rounds. Furthermore, we investigate the influence of the network topology on algorithm performance and derive an explicit relationship between the convergence rate and the graph condition number. Numerical simulations are presented to validate the effectiveness of the proposed approach.

\end{abstract}

\section{Introduction}

Distributed intelligent systems, such as drone swarms, smart grids, and cyber-physical systems, have been widely studied across control theory, signal processing, and machine learning \cite{magnusbook}-\cite{Rabbat2010}. These systems are typically represented mathematically as networks of interacting agents, where each node corresponds to an agent and edges represent communication links between them.
When tasked with objectives like clustering, collaborative control, or optimization, distributed systems must compute in a decentralized manner. In this process, each agent maintains local information, exchanges messages with neighboring nodes, and jointly solves global mathematical problems \cite{magnusbook}.
Among them, distributed optimization problems are widely studied due to their applicability in modeling many real-world scenarios \cite{Patari2022}-\cite{Shen2022}. In such problems, each node has a local function, and the collective goal is to find a solution that minimizes the sum of all local functions through iterative communication across the network.

The common distributed optimization algorithms are established on consensus algorithms, which ensure consistency across the models of agents. A natural combination of gradient descent algorithms and consensus algorithms results in distributed gradient descent algorithms (DGD) \cite{DSMF}. Further algorithms, incorporating additional states, have been proposed to achieve faster convergence rates, e.g. {linear} convergence, such as distributed primal-dual algorithms (DPD) \cite{OAFS,EXTRA,XY-LCOF} and distributed gradient tracking algorithms (DGT) \cite{MB-GT}.

Distributed fast gradient descent algorithms aim to accelerate convergence by extending Nesterov's seminal work \cite{NEST} to a networked setting \cite{FDGM, DNGA, AADS}-\cite{AGAA}. The Nesterov accelerated gradient method is a centralized optimization technique that enhances convergence by combining momentum with gradient descent. Below, we provide a brief comparison of the convergence performance of these algorithms for strongly convex and convex optimization.

\begin{itemize}
\item [(i)]
For strongly convex optimization, DGD requires $\mathcal{O}\left(\epsilon^{-1}\right)$ iterations to reach $\epsilon$-accuracy \cite{DSMF}, while DPD and DGT exhibit linear convergence rates and requires  $\mathcal{O}\left(\frac{L}{\mu} \ln\left(\epsilon^{-1}\right)\right)$ iterations for $\mu$-strongly convex and $L$-Lipschitz continuous functions \cite{EXTRA}. 
As for Nesterov-accelerated methods, the DGT-based algorithm Acc-DNGD-SC, proposed in \cite{ADNG, AGAA}, attains a faster unbiased linear convergence rate and requires $\mathcal{O}\left(\left(\frac{L}{\mu}\right)^{\frac{5}{7}} \ln\left(\epsilon^{-1}\right)\right)$ iterations to achieve $\epsilon$-accuracy.
\item [(ii)] For convex optimization, DGD achieves a convergence rate of $\mathcal{O}\left(\frac{1}{k^{2/3}}\right)$ \cite{DSMF}, while DPD achieve a convergence rate of $\mathcal{O}\left(\frac{1}{k}\right)$ \cite{OAFS}. When the Nesterov method is incorporated, the DGD-based algorithm in \cite{DNGA}, referred to as Acc-B-DGD (Accelerated Biased DGD), uses a constant step size and achieves a biased convergence rate of $\mathcal{O}\left(\frac{1}{k^2}\right)$. In \cite{FDGM}, the D-NG algorithm with decaying step sizes achieves a convergence rate of $\mathcal{O}\left(\frac{\log k}{k}\right)$. With additional communication, the D-NC algorithm proposed in \cite{FDGM} achieves a convergence rate of $\mathcal{O}\left(\frac{1}{k^2}\right)$. In \cite{Xu2020,DPDN}, the Nesterov method is integrated into DPD,  and the resulting DPDN achieves a convergence rate of $\mathcal{O}\left(\frac{1}{k^{2}}\right)$.
\end{itemize}

Notably, although DGD-based algorithms generally exhibit slower convergence rates compared to those based on DPD and DGT, they remain widely studied and have significant application potential due to their simple algorithmic structure and lower communication state requirements \cite{chen2025distributed,FDL}. {
This motivates us to investigate whether DGD can be accelerated without
introducing additional communicated states or communication rounds.
The central objective of this work is therefore to improve its
convergence while preserving its low per-iteration communication cost.
}

\color{black}

In this paper, we introduce distributed Nesterov gradient descent algorithms by exploring their continuous-time approximations. Our primary focus is twofold: first, to apply insights from control theory to analyze the distributed Nesterov gradient descent algorithm; and second, to leverage conditions derived from the continuous-time flow to inform the design of new discrete-time algorithms that guarantee faster convergence rates for both strongly convex and convex optimization problems. Additionally, we investigate the impact of graph topology on the proposed algorithms. The main contributions of this work are summarized as follows:

\begin{itemize}
\item [(i)] We derive a general continuous-time formulation for distributed Nesterov gradient descent algorithms. This formulation not only provides a deeper understanding of the relationships between key parameters, such as the consensus coefficient, friction coefficient, and gradient coefficient, but also facilitates the extension of Nesterov acceleration techniques to distributed settings.
\item [(ii)] By exploring various parameter settings within this continuous-time framework, we propose two new algorithms. Compared to other DGD-based algorithms, the proposed DNGD-SC requires fewer iterations, specifically $\mathcal{O}\left(\frac{\ln \epsilon^{-1}}{\sqrt{\epsilon}}\right)$, to achieve $\epsilon$-accuracy for strongly convex optimization, and the proposed DNGD-C achieves a faster convergence rate of $\mathcal{O}\left(\frac{1}{k}\right)$ for convex optimization. {
Notably, both algorithms preserve the per-iteration communication
structure of standard DGD, requiring each agent to transmit only one
\(d\)-dimensional state without any additional communication rounds.
} The above convergence rates are compared with other DGD-based algorithms in Table \ref{tab:1}.
\end{itemize}

\begin{table}[h]
    \centering
    \caption{Comparison between our work and corresponding results in literature. $\epsilon$: the desired accuracy; $k$: the iteration in discrete-time algorithms.}
    \label{tab:1}
    \begin{subtable}[t]{0.8\textwidth}
    \centering
    
    \caption{Strongly Convex Optimization}
    \begin{tabular}{c|c|c}
        \hline
        \textbf{Algorithm} & \textbf{Iterations for $\epsilon$} & \textbf{Extra Cost} \\
        \hline
        DGD \cite{DSMF}/D-NG \cite{FDGM} & $\mathcal{O}(\epsilon^{-1})$ & --  \\
        D-NC \cite{FDGM} & $\mathcal{O}(\epsilon^{-2})$ & Multiple Communication  \\
        DNGD-SC & $\mathcal{O}\left(\frac{\ln\epsilon^{-1}}{\sqrt{\epsilon}}\right)$ & -- \\
        \hline
    \end{tabular}
\end{subtable}
\medskip
 \centering
\begin{subtable}[h]{0.8\textwidth}
\centering
\caption{Convex Optimization}
\begin{tabular}{c|c|c}
\hline
\textbf{Algorithm} & \textbf{Convergence Rate} & \textbf{Extra Cost} \\
\hline
DGD \cite{DSMF} & $\mathcal{O}(k^{-2/3})$ & --- \\
D-NG \cite{FDGM} & $\mathcal{O}\left(\dfrac{\log k}{k}\right)$ & --- \\
D-NC \cite{FDGM} & $\mathcal{O}(k^{-2})$ & Multiple Communication \\
Acc-B-DGD \cite{DNGA} & $\mathcal{O}(k^{-2})$ & Biased Convergence \\
DNGD-C & $\mathcal{O}(k^{-1})$ & --- \\
\hline
\end{tabular}
\end{subtable}

\end{table}

 Our work builds on \cite{ADEF}--\cite{large} by extending their analysis to distributed optimization. Prior studies focused on continuous-time formulations of the centralized Nesterov method to guide discrete-time parameter design. In particular, \cite{ADEF} proposed a second-order flow for Nesterov acceleration, identifying the friction coefficient $\frac{\alpha}{t}$ as key to convergence, with $\alpha = 3$ being optimal. Subsequent works \cite{large} and \cite{ROCO} further characterized the regimes $\alpha > 3$ and $0 < \alpha \leq 3$, respectively. In addition, there has also been a wealth of research on distributed continuous flows in various domains  \cite{DPDN}--\cite{LW_SUR,Ren2026}.


\color{black}


\emph{\bf Notation.}
In this paper, $\|\cdot \|$ denotes the Euclidean norm. The notation $\mathbf{1}_n$ ($\mathbf{0}_n$) represents the column vector of ones (zeros), and $\mathbf{I}_n$ denotes the identity matrix. The symbol $\otimes$ represents the Kronecker product. For a differentiable function, $\nabla(\cdot)$ denotes its gradient. For a matrix $\mathbf{A}$, its column space is denoted by $\col(\mathbf{A})$.

\section{Problem Definition}
\label{sec.pro}

\subsection{Distributed Optimization}

 Consider a network of agents indexed by $\mathbb{V}=\left\{1,2,\dots,n\right\}$, where each agent $i\in\mathbb{V}$ holds an objective function $f_i:\mathbb{R}^d\rightarrow \mathbb{R}$. The agents aim to solve the following system-level optimization problem 
\begin{equation}
    \label{eq:DO}
    \ba
    \mathrm{min}&\  \sum_{i\in\mV} f_i\left(\xb_i\right),\\
    \mathrm{s.t.}&\ \xb_i=\xb_j, \quad\forall i,j\in \mV.
    \ea
\end{equation}
Each local gradient  $\nabla f_i$ is assumed to be globally $L$-Lipschitz continuous and it is assumed that $f_i>-\infty$.

As each agent only has information about its own objective function, to solve such a distributed optimization problem \eqref{eq:DO}, a communication network is usually required to transmit messages.
The network communication can be modeled as a graph $\mathrm{G}(\mathbb{V}, \mathbb{E})$, where $\mathbb{E}$ represents the set of edges, and each agent communicates only with its neighbors. Let $[a_{ij}] \in \mathbb{R}^{n \times n}$ denote the weight matrix associated with the graph $\mathrm{G}$, where $a_{ij} > 0$ if $\{j, i\} \in \mathbb{E}$ and $a_{ij} = 0$ if $\{j, i\} \notin \mathbb{E}$. Let $\Lb$ denote the Laplacian matrix of the graph $\mathrm{G}$, where $[\Lb]_{ij} = -a_{ij}$ for all $i \neq j$, and $[\Lb]_{ii} = \sum_{j=1}^n a_{ij}$ for all $i \in \mathbb{V}$. In addition, denote the neighbor set of agent $i$ as $\mathbb{N}_i$, satisfying $j\in\mathbb{N}_i$ if and only if $[\Lb]_{ij}\neq0$ for all $i,j\in\mV$. For simplicity, we assume that the graph $\mathrm{G}$ is undirected, connected, and time-invariant. This implies that the Laplacian matrix is symmetric and positive semi-definite, with eigenvalues $\lambda_i$ for $i \in \mathbb{V}$, ordered in ascending order as $0 = \lambda_1 < \lambda_2 \leq \dots \leq \lambda_n$. We define 
the \emph{condition number} of $\Lb$ as $\kappa:=\frac{\lambda_n}{\lambda_2}\geq1$. The larger $\kappa$ is, the worse the connectivity of the graph be. 
Furthermore, we have $\Lb \mathbf{1}_n = \mathbf{0}_n$. Please see \cite{magnusbook} for more details and explanations.

\subsection{Distributed Nesterov Gradient Descent Flow}

Algorithms for solving Problem \eqref{eq:DO} have been extensively studied in the literature. A general form of the distributed Nesterov gradient descent algorithm can be {derived} by
\begin{equation}
\label{eq:DGDO-N-dis}
\ba
{\xb_{i,k+1}} &= \yb_{i,k} - \eta \sum_{j \in \mathbb{N}_i} \Lb_{ij} \xb_j + \hat\beta(k, \eta) \nabla f_i({\yb_{i,k}}), \\
\yb_{i,k} &= \xb_{i,k} + \hat{\alpha}(k, \eta) (\xb_{i,k} - \xb_{i,k-1}), \quad i \in \mathbb{V},
\ea
\end{equation}
where \( \eta \), \( \hat{\alpha} \), and \( \hat{\beta} > 0 \) are the consensus coefficient, the friction coefficient, and the gradient coefficient, respectively, and the initial condition is given by \( \xb_{i,0} = \yb_{i,0} \). Notably, the form of \eqref{eq:DGDO-N-dis} is commonly adopted in specific distributed Nesterov gradient descent algorithms in the literature \cite{FDGM, DNGA}.

Next, we derive the continuous-time flow from \eqref{eq:DGDO-N-dis}. The relationship between the parameters in both the discrete-time and continuous-time formulations will be established, providing a foundation for the subsequent design process.

Inspired by the approach in \cite{ADEF}, we introduce the ansatz \( \Xb(t) = \Xb(k \sqrt{\eta}) \approx \xb_k \), where \( \sqrt{\eta} > 0 \) represents the time interval. With this assumption, we have \( \xb_{k+1} \approx \Xb(t + \sqrt{\eta}) \). As \( \eta \to 0 \), the following approximation holds from Taylor expansion,
\begin{equation}
\label{Talor}
\ba
\frac{\xb_{i,k+1} - \xb_{i,k}}{\sqrt{\eta}} &= \dot{\Xb}_i(t) + \sqrt{\eta} \frac{1}{2} \ddot{\Xb}_i(t) + o(\sqrt{\eta}), \\
\frac{\xb_{i,k} - \xb_{i,k-1}}{\sqrt{\eta}} &= \dot{\Xb}_i(t) - \sqrt{\eta} \frac{1}{2} \ddot{\Xb}_i(t) + o(\sqrt{\eta}), \\
\eta \nabla f(\yb_{i,k}) &= \eta \nabla f_i(\Xb(t)) + o(\sqrt{\eta}).
\ea
\end{equation}
Substituting \eqref{Talor} into \eqref{eq:DGDO-N-dis}, we obtain the following continuous-time equation,
\[
\ba
&\sqrt{\eta} \dot{\Xb}_i + \frac{1}{2} \eta \ddot{\Xb}_i \\
= &\hat{\alpha}\left(\frac{t}{\sqrt{\eta}}, \eta\right)\left( \sqrt{\eta} \dot{\Xb}_i - \eta \frac{1}{2} \ddot{\Xb}_i \right) + \hat{\beta}\left(\frac{t}{\sqrt{\eta}}, \eta\right) \nabla f_{{i}}(\Xb_i) \\
&- \eta \sum_{j \in \mathbb{N}_i} \Lb_{ij} \xb_j + o\left(\sqrt{\eta}\right).
\ea
\]
Comparing the coefficients with respect to \( \eta \), if

\begin{equation}
     \label{eq:alphabeta}
     \ba
\alpha(t) = \lim_{\eta \to 0} \left( \frac{1 - \hat{\alpha}\left( \frac{t}{\sqrt{\eta}}, \eta \right)}{\sqrt{\eta}} \right), \ \ \beta(t) = \lim_{\eta \to 0} \left( \frac{\hat{\beta}\left( \frac{t}{\sqrt{\eta}}, \eta \right)}{\eta} \right)
\ea
\end{equation}
exist, we can obtain the following distributed Nesterov gradient descent flow,

\begin{equation}
\label{eq:DGDO-N}
\ba 
         \ddot{\Xb}_i + \alpha(t) \dot{\Xb}_i + \sum_{j \in \mathbb{N}_i} \Lb_{ij} \xb_j + \beta(t) \nabla f_i(\Xb_i) = 0, \quad i \in \mathbb{V},
\ea
\end{equation}
where the initial condition is \( \dot{\Xb}_i(0) = 0 \).

We can also construct new discrete-time distributed Nesterov gradient descent algorithms from a given \( \alpha(t) \) and \( \beta(t) \) in the flow \eqref{eq:DGDO-N}, by letting

\begin{equation}
     \label{eq:alphabeta2}
     \ba
\hat{\alpha}(k, \eta) &= 1 - \alpha(k \sqrt{\eta}) \sqrt{\eta} + o(\sqrt{\eta}), \\
\hat{\beta}(k, \eta) &= \beta(k \sqrt{\eta}) \eta + o(\eta).
\ea
\end{equation}

\subsection{Problem of Interest}

From the above analysis, the following problem will be the focus of this paper.
How can we leverage insights from control theory to interpret the distributed Nesterov gradient descent algorithm, and how do the conditions derived from the corresponding continuous-time flow \eqref{eq:DGDO-N} inspire the design of new discrete-time algorithms with provably faster convergence rates in both strongly convex and convex optimization problems?

\color{black}





\section{{Strong Convexity with Fixed Coefficients}}
\label{sec.fix}

In this section, we  design the parameters in the flow \eqref{eq:DGDO-N} for strongly convex optimization problems, establish the relationship between convergence rate and parameters, and then propose a discrete-time algorithm to achieve faster convergence.

First, we introduce the following assumption on the strongly convex objective functions.  

\medskip

\begin{assumption}
    \label{ass-strcon}
    The objective function $\sum_{i\in\mV}f_i$ is $\mu$-strongly convex. Specifically, for some $\mu>0$, $f(\Xb)$ satisfies
    \[
    f(\Xb')\geq f(\Xb)+\nabla f(\Xb)^\top(\Xb'-\Xb)+\frac{\mu}{2}\|\Xb'-\Xb\|^2
    \]
    for all $ \Xb,\Xb'\in\mathbb{R}^d$, where $f: = \sum_{i\in\mV} f_i$. \hfill $\square$
\end{assumption}
If Assumption \ref{ass-strcon} holds, there exists a unique optimal solution $\Xb^\ast\in\mathbb{R}^d$ such that 
$
\nabla f(\Xb^\ast)=0,\quad f(\Xb^\ast)=f^\ast,
$
where $f^\ast$ denotes its optimal value.  

{
Under this assumption, we will provide convergence analysis of flow \eqref{eq:DGDO-N}.}
For strongly convex optimization we consider setting the friction coefficient $\alpha(t)$ to a constant value following \cite{ADEF,NEST}. To further accelerate convergence, we select the gradient coefficient $\beta(t)$ in the flow \eqref{eq:DGDO-N} to be a constant $\beta$, instead of letting it decay to zero. With this choice, since $\nabla f_i(\Xb^\ast)\neq 0$, the resulting algorithms cannot ensure strict unbiasedness.

In this setting, we know $\Xb_i(t)$ generated by flow \eqref{eq:DGDO-N} satisfies biased linear convergence, as stated in the following lemma.
\begin{lemma}
\label{lem-sc-2}
Let Assumption \ref{ass-strcon} hold. Assume it holds $\Xb_i(0)=\Xb_j(0)$ for all $i,j\in\mathbb{V}$. Then, for $\Xb_i(t)$ generated by flow \eqref{eq:DGDO-N} with fixed parameters 
\begin{equation}
    \label{eq:alpha}
    \ba
&\alpha(t) =\alpha=2\sqrt{\min\left\{\lambda_2,\frac{\beta\mu}{n}\right\}},\\
&\beta(t)=\beta\in \left(0, \min\left\{\frac{\lambda_2}{4L^2},\frac{\mu\lambda_2}{16n^2L^2}\right\}\right),
\ea
\end{equation} it holds that
\[
    \ba
f(\Xb_i(t)) - f^\ast\leq Be^{-\sqrt{\min\left\{\lambda_2,\frac{\beta\mu}{n}\right\}}t}+\beta {D}
\ea
\]
for some constants $B,D>0$ and  for all $i\in\mathbb{V}$. 
\end{lemma}

With the coefficient values in \eqref{eq:alpha} established, we are now ready to propose the discrete-time algorithm and analyze its corresponding coefficients.
 We recall the algorithm proposed in \cite{NEST} as follows. 

\begin{equation}
    \label{eq:cla_nes2}
    \ba
    \xb_{k+1} &= \yb_{k} - \eta \nabla f(\yb_{k}), \\
    \yb_{k} &= \xb_{k} + \frac{1 - \sqrt{\mu \eta}}{1 + \sqrt{\mu \eta}} (\xb_{k} - \xb_{k-1}),
    \ea
\end{equation}
where $\eta > 0$ and the initial condition is $\xb_0 = \yb_0$.

In \eqref{eq:cla_nes2}, if $f$ is $\mu$-strongly convex, $\nabla f$ is $L$-Lipschitz continuous with optimal value $f^\ast$, and $\eta = \frac{1}{L}$, the value $f(\xb_k) - f^\ast$ will converge to zero linearly \cite{NEST}. Compared with \eqref{eq:DGDO-N-dis}, we have $\hat\alpha(k, \eta) = \frac{1 - \sqrt{\mu \eta}}{1 + \sqrt{\mu \eta}}$ and $\hat{\beta}(k, \eta) = \eta$ in \eqref{eq:cla_nes2}. Substituting these into \eqref{eq:alphabeta}, we get $\alpha(t) = 2 \sqrt{\mu}$ and $\beta = 1$. We then obtain the following flow
\begin{equation}
    \label{eq:scODE}
    \ba
         \ddot{\Xb} + 2 \sqrt{\mu} \dot{\Xb} + \nabla f(\Xb) &= 0,
    \ea
\end{equation}
where the initial condition is $\dot{\Xb}(0) = 0$. 

\color{black}
Comparing \eqref{eq:cla_nes2} and flow \eqref{eq:scODE}, we have
\[
    \hat\alpha = \frac{2\alpha\sqrt{\eta} - \alpha^2\eta}{2\alpha\sqrt{\eta} + \alpha^2\eta}, \quad \hat\beta = \beta\eta.
\]
Substituting the above results and \eqref{eq:alpha} into \eqref{eq:DGDO-N-dis}, we obtain the following Distributed Nesterov Gradient Descent algorithm for Strongly Convex optimization (DNGD-SC),
\begin{equation}
    \label{eq:DNGD-SC}
    \ba
    \xb_{i,k+1} &= \yb_{i,k} -  \eta \sum_{j \in \mathbb{N}_i} \Lb_{ij} \yb_{j,k} - \eta \beta \nabla f_i(\yb_{i,k}), \\
    \yb_{i,k} &= \xb_{i,k} + \frac{2\alpha\sqrt{\eta} - \alpha^2 \eta}{2\alpha\sqrt{\eta} + \alpha^2 \eta} (\xb_{i,k} - \xb_{i,k-1}), \quad i \in \mV,
    \ea
\end{equation}
where \(\eta, \alpha, \beta > 0\) and the initial condition is \(\xb_{i,0} = \yb_{i,0}\). We have the following result for DNGD-SC \eqref{eq:DNGD-SC}.

{
\begin{theorem}\label{thm-sc-dis}
Let Assumption~\ref{ass-strcon} hold, and suppose that
\(\xb_{i,0}=\xb_{j,0}\) for all \(i,j\in\mathbb{V}\).
Consider \(\xb_{i,k}\) generated by DNGD-SC~\eqref{eq:DNGD-SC}.
For a prescribed accuracy \(\epsilon>0\), choose \(
\beta=c_\beta\epsilon,
\)
where \(c_\beta>0\) is independent of \(\epsilon\) and satisfies
\[
c_\beta\leq\frac{\lambda_2}{2D},
\qquad
c_\beta\epsilon\leq
\min\left\{
\frac{\lambda_2}{4L^2},
\frac{\mu\lambda_2}{16n^2L^2}
\right\},
\]
and let
\(
\eta=\frac{1}{\lambda_n+\beta\sqrt{n}L}\), \(
\alpha=2\sqrt{\min\left\{\lambda_2,\frac{\beta\mu}{n}\right\}}.
\)
Then \(f(\xb_{i,k})-f^\ast\leq\epsilon\) for all
\(k>k_0(\epsilon,\beta)\), where
\[
k_0(\epsilon)
=
\frac{\ln\left(\frac{\epsilon}{2B}\right)}
{\ln\left(
1-
\sqrt{
\frac{
\min\left\{\lambda_2,\frac{c_\beta\epsilon\mu}{n}\right\}
}{
\lambda_n+c_\beta\epsilon\sqrt{n}L
}
}
\right)}=\mathcal{O}\left(
\sqrt{\kappa}\,
\frac{\ln\epsilon^{-1}}{\sqrt{\epsilon}}
\right).
\]
Here, \(B>0\) and \(D>0\) are constants determined
by the initial value and can be seen in the proof.
\hfill$\square$
\end{theorem}
}

\color{black}

The proof of Theorem \ref{thm-sc-dis} can be seen in Appendix \ref{app:thm2}. {Here, the identical-initialization condition \(\xb_{i,0}=\xb_{j,0}\) can be satisfied by prescribing
\(\xb_{i,0}=\mathbf{0}_d\).}
Theorem \ref{thm-sc-dis} establishes the required number of iterations to reach \(\epsilon\)-accuracy, which is \(\mathcal{O}\left( \frac{\ln \epsilon^{-1}}{\sqrt{\epsilon}} \right)\) and proportional to \(\sqrt{\kappa}\).

\begin{remark}
\label{remark:common}
    It is important to note that the biased convergence of the algorithm arises due to the {inconsistency between the local optimum and the global optimum for each node in Problem \eqref{eq:DO}}. Therefore, consider the global optimal solution \(\Xb^\ast\) minimizes each local function \(f_i\) of the nodes. A typical example is given by
    \[
    \ba
    \min \sum_{i \in \mV} f_i(\Xb) = \sum_{i \in \mV} \|\Hb_i^\top \Xb - b_i\|^2,
    \ea
    \]
    where \(\Hb_i \in \mathbb{R}^n\) and \(b_i \in \mathbb{R}\). If \(\mathbf{H} := [\mathbf{H}_1 \, \dots \, \mathbf{H}_n]\) has full rank, then \(\Xb^\ast = \mathbf{H}^{-1} \bb\) is the minimal solution of \(f = \sum_{i \in \mV} f_i(\Xb)\) and each \(f_i\). 
    {In this case, following \eqref{eq:appa}, we set $\beta = \frac{\lambda_2 n}{\mu}$. Then, DNGD-SC \eqref{eq:DNGD-SC} achieves unbiased linear convergence and requires 
$
k_0 = \mathcal{O}\!\left(\sqrt{\frac{L}{\mu}} \ln \epsilon^{-1}\right)
$ iterations to attain $\epsilon$-accuracy. This result is consistent with the convergence rate of the centralized Nesterov method for strongly convex problems \cite{NEST}. It is faster than existing distributed algorithms \cite{ADNG,AGAA, EXTRA}, which, however, are designed for more general settings in which inconsistencies between local optima and the global optimum may arise.}
\hfill$\square$
\end{remark}

\section{{Convexity with Time-varying Coefficients}}
\label{sec.tim}
In this section, we consider the following assumptions for the objective functions. 

\begin{assumption}\label{ass-convex}
    The objective function \( f_i \) is convex, i.e.,
    \[
    f_i(\Xb'_i) \geq f_i(\Xb_i) + \nabla f_i(\Xb_i)^\top (\Xb'_i - \Xb_i)
    \]
    for all \( \Xb'_i, \Xb_i \in \mathbb{R}^d \) and \( i \in \mathbb{V} \). In addition, the optimal solution set of Problem \eqref{eq:DO} is compact and non-empty.
    \hfill $\square$
\end{assumption}

If Assumption \ref{ass-convex} holds, there exists an optimal solution set \( \mathbb{S} \subset \mathbb{R}^d \), such that \( \nabla f(\Xb^\ast) = 0 \) and \( f(\Xb^\ast) = f^\ast \) for all \( \Xb^\ast \in \mathbb{S} \).

\subsection{{The Flow and Convergence Analysis}}

For convex optimization, we first let the friction coefficient \( \alpha(t) = \frac{r}{t} \) with \( r > 0 \) in flow \eqref{eq:DGDO-N}, which is common for convex optimization using the Nesterov method \cite{ADEF}. Since the local optimal solution in each node differs from the global optimal solution, we choose \( \beta(t) = \frac{1}{t^p} \) with \( p > 0 \) to ensure that consensus is reached and the solution remains unbiased. Specifically, the form of the flow \eqref{eq:DGDO-N} discussed in this section is given by
\begin{equation}
\label{eq:DGDO-N-C}
\ba
    \ddot\Xb_i + \frac{r}{t} \dot\Xb_i + \sum_{j \in \mathbb{N}_i} \Lb_{ij} \Xb_j + \frac{1}{t^p} \nabla f_i(\Xb_i) &= 0, \quad i \in \mathbb{V},
\ea
\end{equation}
where \( r, p > 0 \) and the initial condition is \( \dot{\Xb}_i(0) = 0 \). 

{
Next, we analyze the optimal choices of the parameters 
$r$ and $p$ with respect to the convergence rate of the flow \eqref{eq:DGDO-N-C}, and establish the corresponding results.}

To provide a more comprehensive understanding of the flow \eqref{eq:DGDO-N-C}, we begin by examining the relationship between 
$r$ and $p$ through the following two lemmas.
\medskip
\begin{lemma}
    \label{lem:1}
    Let $f$ be convex. Then, for \(\Xb_i(t)\) generated by the following flow
    \[
    \ba
        \ddot\Xb + \frac{r}{t}{\dot\Xb}_i + \frac{1}{t^p}\nabla f(\Xb) &= 0,
    \ea
    \]
    with \(r + p \geq 3\), \(r \geq 1\), and \(\dot{\Xb}_i(0) = 0\), it holds that
    $
    f(\Xb(t)) - f^\ast = \mathcal{O}\left(\frac{1}{t^{2-p}}\right),
    $
    for all \(i \in \mathbb{V}\), where \(f_i^\ast\) denotes the optimal value of \(f_i\).
    \hfill
\end{lemma}
\medskip

The proof of Lemma \ref{lem:1} can be extended from \cite{AVPO}.
\medskip
\begin{lemma}
    \label{lem:2}
    For \(\Xb_i(t)\) generated by the following flow
    \begin{equation}
    \label{eq:DGDO-N-C_aux}
    \ba
        \ddot{\Xt} + \frac{r}{t}\dot{\Xt} + \Lc\Xt &= 0,
    \ea
    \end{equation}
    with \(r \geq 2\), it holds that
    $
    \Xt^\top \Lc \Xt = \mathcal{O}\left(\frac{1}{t^2}\right).
    $
\end{lemma}
\medskip

Note that \(g^{\frac{r-1}{2}}\) is convex for \(g(\Xt) := \Xt^\top \Lc \Xt\) when \(r \geq 2\). The proof of Lemma \ref{lem:2} can be extended from Theorem 7 in \cite{ADEF}. 

{
Building on these two lemmas, we establish that \( r + p \geq 3 \) and \( r \geq 2 \). For faster convergence, \( p \) should be chosen as small as possible, leading to \( p = 3 - r \) under the constraint \( r \geq 2 \) for the flow \eqref{eq:DGDO-N-C}. Given $p>0$, we have $r\in[2,3)$. 

We are now in a position to propose the following theorem to show the convergence results of flow \eqref{eq:DGDO-N-C}. 
}

\medskip

\begin{theorem}\label{thm-c-con}
Let Assumption \ref{ass-convex} hold. Then  for \( \Xb_i(t) \) generated by flow \eqref{eq:DGDO-N-C} with {\( r \in [2,3) \) and \( p = 3-r \), it holds
\begin{equation}
    \label{eq.thm2_a}
f(\Xb_i(t)) - f^\ast=\mathcal{O}\left(\frac{1}{t^{3-r}}\right)
\end{equation}}
Specially, when $r=2$, it holds
\begin{equation}
    \label{eq.thm2_b}
\ba
    f(\Xb_i(t)) - f^\ast \leq& \frac{2B^2 + B\sqrt{4B^2 + 2\lambda_2 \sum_{i \in \mathbb{V}} \|\Xb_i(0) - \Xb^\ast\|^2}}{t \lambda_2} \\
    & + \frac{1}{2t} \sum_{i \in \mathbb{V}} \|\Xb_i(0) - \Xb^\ast\|^2,
\ea
\end{equation}
for some constant \( B > 0 \) and for all \( i \in \mathbb{V} \) and all \( \Xb^\ast \in \mathbb{S} \).

\end{theorem}
\medskip

{
It follows from \eqref{eq.thm2_a} that the optimal convergence rate is achieved at \(r = 2\) and \(p = 1\). Now we are ready to propose the corresponding discrete-time algorithms.}

\subsection{Nesterov Flow Inspired Discrete-time Algorithm}

Based on flow \eqref{eq:DGDO-N-C} with $r=2,\ p=1$ and \eqref{eq:alphabeta2}, we propose the following Distributed Nesterov Gradient Descent algorithm for Convex optimization (DNGD-C):
\begin{equation}
\label{eq:DNGD-C}
\ba
\xb_{i,k+1}
=&\ \yb_{i,k}
{-}\eta\sum_{{j\in\mathbb{N}_i}}\Lb_{ij}\yb_{j,k}
{-}\frac{\sqrt{\eta}}{k+1}\nabla f_i(\yb_{i,k}),\\
\yb_{i,k}
=&\ \xb_{i,k}
+\frac{k-1}{k+1}(\xb_{i,k}-\xb_{i,k-1}),
\quad i\in\mathbb{V},
\ea
\end{equation}
where \({\eta>0}\) and
\(\xb_{i,0}=\yb_{i,0}\).

For DNGD-C \eqref{eq:DNGD-C}, a theorem similar to Theorem \ref{thm-c-con} can be obtained as follows,
\medskip
\begin{theorem}\label{thm-c-dis}
Let Assumptions \ref{ass-convex} hold. Then with \(0 < \eta \leq \left( \frac{\sqrt{L^2 + 4\lambda_n} - L}{2\lambda_n} \right)^2\), for \(\xb_{i,k}\) generated by DNGD-C \eqref{eq:DNGD-C}, we have
\begin{equation}
\label{eq:concl}
\ba
f(\xb_{i,k}) - f^\ast \leq &\ \frac{2B^2 + B\sqrt{4B^2 + 2\lambda_2 \sum_{i \in \mathbb{V}} \|\Xb_i(0) - \Xb^\ast\|^2}}{k\sqrt{\eta}\lambda_2} \\
& + \frac{\sum_{i \in \mathbb{V}} \|\Xb_i(0) - \Xb^\ast\|^2}{k\sqrt{\eta}},
\ea
\end{equation}
for some constant \(B > 0\) and for all \(i \in \mathbb{V}\) and all \(\Xb^\ast \in \mathbb{S}\). In addition, if \(\eta = \left( \frac{\sqrt{L^2 + 4\lambda_n} - L}{2\lambda_n} \right)^2\), for a fixed \(\lambda_2\), we have
\[
\ba
f(\xb_{i,k}) - f^\ast \leq &\ \mathcal{O}\left( \frac{\sqrt{\kappa}}{k} \right).
\ea
\] 
\end{theorem}
\medskip

The proof of Theorem \ref{thm-c-dis} can be found in Appendix \ref{app:thm4}. From Theorem \ref{thm-c-dis}, we know that DNGD-C provides the convergence rate \(\mathcal{O}\left( \frac{1}{k} \right)\). Moreover, for a larger condition number of the Laplacian matrix, DNGD-C exhibits slower convergence rates.

{
For the case where the objective function is convex and satisfies the property in Remark~\ref{remark:common}, i.e., local optima coincide with global optima, we can directly adopt the parameter design of the centralized Nesterov method, i.e., $p=0$, $r=3$. In this case, the convergence rate of the algorithm will be the classical $O\left(\frac{1}{k^2}\right)$.}

\begin{remark}
The convergence rates established in Theorem~\ref{thm-sc-dis} for \eqref{eq:DNGD-SC} and in Theorem~\ref{thm-c-dis} for \eqref{eq:DNGD-C} do not provide improvements in terms of convergence rate and are even inferior to the Nesterov-accelerated variants of DPD \cite{OAFS,AADS,DPDN} and DGT \cite{MB-GT,ADNG, AGAA}. Nevertheless, the DGD framework considered here entails lower communication overhead than DPD and DGT. At each iteration, both DPD and DGT require transmitting an auxiliary variable with the same dimension as $\xb$, leading to roughly twice the communication cost of DGD \footnote{Though some DPD variants eliminate this extra exchange \cite{XY-LCOF}, they do so by introducing additional inter-node initial conditions, reducing robustness.}, {while their accelerated variants require
even more communication.}
As a result, though the DPD- and DGT-based methods have faster convergence rates, the communication cost required to achieve a certain accuracy in convergence is greater than that of the proposed \eqref{eq:DNGD-SC} and \eqref{eq:DNGD-C}, a benefit that is also confirmed by the simulations in Section~\ref{sec.num}.
 \hfill $\square$
\end{remark}
\color{black}


\section{Numerical Simulations}
\label{sec.num}

\subsection{Benchmark against Existing Methods}
\label{sec:V-A}
In this paper, we have proposed two algorithms: DNGD-SC \eqref{eq:DNGD-SC} and DNGD-C \eqref{eq:DNGD-C}. In the following, we compare the two proposed algorithms DNGD-SC \eqref{eq:DNGD-SC} and DNGD-C \eqref{eq:DNGD-C} in convergence performances with some existing methods.

First, we consider the strongly convex optimization problem 
\eqref{sim2}:

\begin{equation}
\label{sim2}
\min_{\xb} \sum_{i \in \mathbb{V}} f_i(\xb)
= \sum_{i \in \mathbb{V}}  \|\mathbf{H}_i^\top \xb - \bb_i\|^2  ,
\end{equation}
where $\Hb_i=[\Hb_{i,1},\Hb_{i,2}]$, \(\mathbf{H}_i, \mathbf{H}_i' \in \mathbb{R}^d\) and \(\bb_i\in \mathbb{R}^2\) are randomly generated. The matrices \(\mathbf{H} = [\mathbf{H}_{i,1}, \dots, \mathbf{H}_{n,1}]\) and \(\mathbf{H}' = [\mathbf{H}_{1,2}, \dots, \mathbf{H}_{n,2}]\) are constructed to have full row rank.
{{
Problem \eqref{sim2} is a distributed linear least‑squares problem, which is widely used in collaborative estimation tasks.}}
\color{black}
Assumption~\ref{ass-strcon} holds for \eqref{sim2}, and the minimizer of \(\sum_{i \in \mathbb{V}} f_i(\xb)\) differs from those of \(f_i(\xb)\), avoiding the setting in Remark~\ref{remark:common}.

We compare DNGD-SC \eqref{eq:DNGD-SC} with existing methods in reaching an accuracy of \(10^{-3}\), as reported in Table~\ref{tab:1}(a). For completeness, we also include DGD-based methods (DGD, D-NG, D-NC, and Acc-B-DGD), as well as DGT, its accelerated variant Acc-DNGD-SC and DPD.

Fig.~\ref{fig.sc} shows the evolution of \(\sum_{i \in \mathbb{V}} (f(\xb_{i,k}) - f^\ast)\) versus the total number of communicated scalars. DNGD-SC achieves linear convergence and outperforms DGD-based methods, though with a bias, consistent with Theorem~\ref{thm-sc-dis}. In contrast, DPD, DGT and Acc-DNGD-SC converge linearly without bias. Nevertheless, due to lower communication cost per iteration, DNGD-SC reaches the target accuracy faster despite the faster asymptotic rates of DPD and DGT-based methods.

\color{black}

\begin{figure}[htp]
    \centering
   \includegraphics[width=0.48\textwidth]{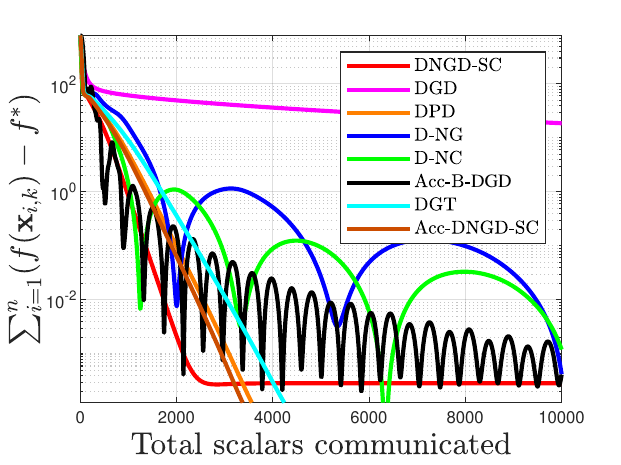}
    \caption{{The values \( \sum_{i \in \mathbb{V}} (f(\mathbf{x}_{i,k}) - f^\ast) \) for Problem \eqref{sim2} generated by different algorithms.}}
    \label{fig.sc}
\end{figure}

Next, we consider the following convex optimization problem \eqref{sim1}:
{
\begin{equation}
\label{sim1}
\min  \sum_{i \in \mathbb{V}} f_i(\mathbf{x}) = \sum_{i \in \mathbb{V}}  \log\Bigl(1 + \exp\bigl(-y_{i} \mathbf{x}^\top \mathbf{a}_{i}\bigr)\Bigr),
\end{equation}
where $a_i\in\mathbb{R}^d$ and $y_i\in\{-1,1\}$. Problem \eqref{sim1} is a distributed logistic regression problem, which is widely used in binary classification tasks.
 It can be easily proven that Assumption \ref{ass-convex} holds for Problem \eqref{sim1}. 
We compare the performance of DNGD-C \eqref{eq:DNGD-C} with several algorithms from the literature, including DGD-based methods (DGD, D-NG, D-NC), as well as DPD, DGT, and their accelerated variants (DPDN and Acc-DNGD-NSC). The results are shown in Fig.~\ref{fig.con}. DNGD-C achieves a faster convergence rate of \( \mathcal{O}\!\left(\frac{1}{k}\right) \) compared to DGD-based methods such as DGD and D-NG, confirming Theorem~\ref{thm-c-dis}. Compared to D-NC, Acc-DNGD-NSC, and DPD- and DGT-based methods, DNGD-C avoids extra communication overhead, resulting in faster convergence in the stages shown in the figure, although it is slower in the long term.
}

\begin{figure}[htp]
    \centering

   \includegraphics[width=0.48\textwidth]{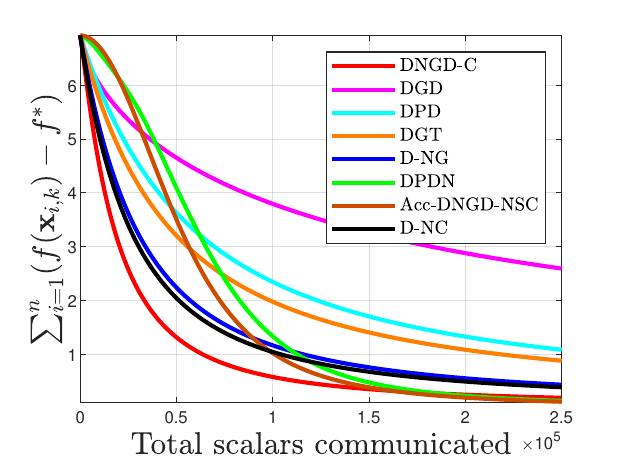}

    \caption{{The values $\sum_{i\in\mV}(f(\mathbf{x}_{i,k})-f^\ast)$ for Problem \eqref{sim1} generated by different algorithms.}}
    \label{fig.con}
\end{figure}

\section{Conclusion}
\label{sec.con}
In this paper, we propose and analyze a general flow that corresponds to the distributed Nesterov gradient descent algorithm. By investigating the flow under different parameter designs, we derive two discrete-time algorithms, DNGD-SC and DNGD-C. These algorithms  require \( \mathcal{O}\left(\frac{\ln \epsilon^{-1}}{\sqrt{\epsilon}}\right) \) iterations to reach \( \epsilon \)-accuracy for strongly convex optimization and achieve a convergence rate of \( \mathcal{O}\left(\frac{1}{k}\right) \) for convex optimization. 

Moving forward, we aim to further explore the potential of combining higher-order algorithms with the Nesterov method, particularly through the lens of continuous-time flows. Additionally, the development of new discrete-time algorithms, using techniques such as Euler discretization, will be a key direction for future research.

\appendix

\section{Proof of Lemma  \ref{lem-sc-2}}
\label{sec.hfb}

Flow \eqref{eq:DGDO-N} can be written compactly as
\begin{equation}
    \label{eq:DGDO-N-SC-ii_com}
    \ba
    \ddot{\Xt} + 2\sqrt{\mu_\theta}{\dot\Xt} + \nabla\Theta_\beta = 0,
    \ea
\end{equation}
where $\Xt(t) := [\Xb_1(t); \dots; \Xb_n(t)]$, $\Lc := \Lb \otimes \mathbf{I}_d$, $\mu_{\theta} := \min\left\{\lambda_2, \frac{\beta\mu}{n}\right\}$, and $\Theta_\beta(\Xt) := \frac{1}{2} \Xt^\top \Lc \Xt + \beta \ft(\Xt)$ with $\ft(\Xt) := \sum_{i \in \mV} f_i(\Xb_i)$.

Note that $\Lb$ is semi-definite with the null space $\col(\mathbf{1}_n)$, which implies
\begin{equation}
\label{eq:LV}
    \Lc \Vb = \mathbf{0}
\end{equation}
for all $\Vb \in \col(\mathbf{1}_n \otimes \mathbf{I}_d)$.
For the function $\Theta_\beta(\Xt)$, we have the following bounds
\[
\ba
&\Theta_\beta(\Xt + \Vb) \\
\geq &\frac{1}{2} \Xt^\top \Lc \Xt + \beta \ft(\Xt + \Vb) \\
\geq & \Theta_\beta(\Xt) + \Vb^\top \nabla \Theta_\beta(\Xt) + \frac{\beta \mu}{2n} \|\Vb\|^2 
\ea
\]
for all $\Vb \in \col(\mathbf{1}_n \otimes \mathbf{I}_d)$ by \eqref{eq:LV} and Assumption \ref{ass-strcon}, and
\[
\ba
&\Theta_\beta(\Xt + \Vb) \\
\geq & \frac{1}{2} (\Xt + \Vb)^\top \Lc (\Xt + \Vb) + \beta \ft(\Xt + \Vb) \\
\geq & \Theta_\beta(\Xt) + \Vb^\top \nabla \Theta_\beta(\Xt) + \frac{\lambda_2}{2} \|\Vb\|^2
\ea
\]
for all $\Vb \notin \col(\mathbf{1}_n \otimes \mathbf{I}_d)$.
This shows that $\Theta_\beta$ is $\mu_\theta$-strongly convex. 

Next, we introduce the following lemma from \cite{NEST}.
\begin{lemma}\label{thm-scODE}
Let $f$ be $\mu$-strongly convex. Then, for $\Xb(t)$ generated by the flow \eqref{eq:scODE}, it holds
\[
f(\Xb(t)) - f^\ast = \mathcal{O}(e^{-\sqrt{\mu} t}).
\]
\end{lemma}
According to Lemma \ref{thm-scODE}, assuming $\Theta_\beta^\ast$ is the minimal value of $\Theta_\beta$, we know that the trajectory $\tilde{\Xb}(t)$ generated by \eqref{eq:DGDO-N-SC-ii_com} satisfies
\begin{equation}
\label{eq:sc_dotv}
\ba
&\Theta_\beta(\Xt(t)) - \Theta_\beta^\ast \leq 2 e^{-\sqrt{\mu_\theta} t} (\Theta_\beta(\Xt(0)) - \Theta_\beta^\ast).
\ea
\end{equation}

Assume $\tilde{\Xb}_\beta^\ast \in \mathbb{R}^{nd}$ satisfies $\Theta_\beta(\tilde{\Xb}_\beta^\ast) = \Theta_\beta^\ast$. Since $f_i > -\infty$, we know that $\ft$ has a minimal value $\ft^\ast$, and consensus between nodes, i.e., $\Xb_1(t)=\Xb_2(t)=\dots=\Xb_n(t)$, is not required for this value. Without loss of generality, we assume $\ft(\tilde{\Xb}_{\rm nc}^\ast) = \ft^\ast$ for some $\tilde{\Xb}_{\rm nc}^\ast \in \mathbb{R}^{nd}$.
Since $\nabla f_i$ is $L$-Lipschitz continuous, it follows that $\nabla \tilde{f}$ is $\sqrt{n} L$-Lipschitz continuous. Hence, we conclude that
\begin{equation}
    \label{eq:fact2}
\ba
&\| \nabla \ft(\Xt(t))\|^2 \\
\leq & 2n L^2 \|\Xt(t) - \Xt_\beta^\ast\|^2 + 2n L^2 \|\Xt_{\rm nc}^\ast - \Xt_\beta^\ast\|^2 \\
\leq & e^{-\sqrt{\mu_\theta} t} \frac{8n L^2}{\mu_\theta} (\Theta_\beta(\Xt(0)) - \Theta_\beta^\ast) + D
\ea
\end{equation}
for $D := 2n L^2 \sup_{\beta\geq0}\|\Xt_{\rm nc}^\ast - \Xt^\ast_\beta\|^2$ and $\Xt^\ast := \mathbf{1}_n \otimes \Xb^\ast$, where the first inequality uses the fact that $\nabla \tilde{f}(\Xt_{\rm nc}^\ast) = \mathbf{0}$ and the last inequality follows from the fact 
\[
\Theta_\beta(\Xt(t)) - \Theta_\beta^\ast \geq \frac{\mu_\theta}{2} \|\Xt(t) - \Xt_\beta^\ast\|^2
\]
and  \eqref{eq:sc_dotv}. 

Next, we compute
\begin{equation}
\label{eq:fact1}
\ba
    &\ft(\mathbf{1}_n \otimes \Xb_i(t)) - \ft(\Xt(t)) \\
\leq & \nabla \ft(\Xt(t))^\top (\mathbf{1}_n \otimes \Xb_i(t) - \Xt(t)) + \frac{L^2}{2} \|\mathbf{1}_n \otimes \Xb_i(t) - \Xt(t)\|^2 \\
\leq & 2\beta \|\nabla \ft(\Xt(t))\|^2 + \left(\frac{\lambda_2}{8\beta} + \frac{L^2}{2}\right) \|\mathbf{1}_n \otimes \Xb_i(t) - \Xt(t)\|^2 \\
\leq & \frac{16n L^2 \beta (\Theta_\beta(\Xt(0)) - \Theta_\beta^\ast)}{\lambda_2 \mu_\theta} e^{-\sqrt{\mu_\theta} t} + \beta \frac{D}{\lambda_2} \\
& + \left(\frac{\lambda_2}{8\beta} + \frac{L^2}{2}\right) \xi_i^2(t)
\ea
\end{equation}
for all $i \in \mV$ and $\xi_i(t) := \|\mathbf{1}_n \otimes \Xb_i(t) - \Xt(t)\|$, where the last inequality follows from \eqref{eq:fact2}.
Thus, for all \(i \in \mV\), we have
\begin{equation}
    \label{eq:inter}
    \ba
    &\ft(\mathbf{1}_n \otimes \Xb_i(t)) - f^\ast \\
    \leq & \frac{\Theta_\beta(\Xt(t)) - \Theta_\beta^\ast}{\beta} + \ft(\mathbf{1}_n \otimes \Xt_i(t)) - \ft(\Xt(t)) - \frac{1}{2\beta} \Xt^\top \Lc \Xt \\
    \leq & \left( \frac{\Theta_\beta(\Xt(0)) - \Theta_\beta^\ast}{\beta} + \frac{16n L^2 \beta (\Theta_\beta(\Xt(0)) - \Theta_\beta^\ast)}{\lambda_2 \mu_\theta} \right) e^{-\sqrt{\mu_\theta} t} \\
    & + \beta \frac{D}{\lambda_2} + \left( \frac{L^2}{2} - \frac{\lambda_2}{8\beta} \right) \xi_i^2(t),
    \ea
\end{equation}
where the first inequality follows from the fact that
\[
\ba
&\Theta_\beta^\ast \leq \Theta_\beta(\mathbf{1}_n \otimes \Xb^\ast)\\ =& \beta \tilde{f}(\Xb^\ast) + \frac{1}{2} (\mathbf{1}_n \otimes \Xb^\ast)^\top \Lc (\mathbf{1}_n \otimes \Xb^\ast) = \beta f^\ast
\ea
\]
derived from \eqref{eq:LV}, and the second inequality is obtained using \eqref{eq:fact1} and noting that
\begin{equation}
\label{eq:boundsx}
\Xt(t)^\top \Lc \Xt(t) \geq \frac{\lambda_2}{2} \xi_i^2(t).
\end{equation}

Then, {letting} 
\[
\beta \leq \min\left\{ \frac{\lambda_2}{4L^2}, \frac{\mu \lambda_2}{16n^2 L^2} \right\},
\]
we obtain
\begin{equation}
    \label{eq:obj}
    \ba
    &\ft(\mathbf{1}_n \otimes \Xb_i(t)) - f^\ast \\
    \leq & \frac{2 (\Theta_\beta(\Xt(0)) - \Theta_\beta^\ast)}{\beta} e^{-\sqrt{\mu_\theta} t} + \beta \frac{D}{\lambda_2} \\
    \leq & \left( 2 (\ft(\Xt(0)) - \ft^\ast) + \frac{\Xt(0)^\top \Lc \Xt(0)}{\beta} \right) e^{-\sqrt{\mu_\theta} t} + \beta \frac{D}{\lambda_2},
    \ea
\end{equation}
where the last inequality follows from the observation that 
\[
\Theta_\beta(\Xt(t)) - \Theta_\beta^\ast = \beta \ft(\Xt(0)) - \Theta_\beta^\ast \leq \beta \ft(\Xt(0)) - \beta \ft^\ast.
\]

With the initial condition \(\Xb_i(0) = \Xb_j(0)\) for all \(i,j \in \mV\), we have
\[
    \ba
    &\ft(\mathbf{1}_n \otimes \Xb_i(t)) - f^\ast \\
    \leq & 2 (\ft(\Xt(0)) - \ft^\ast) e^{-\sqrt{\mu_\theta} t} + \beta \frac{D}{\lambda_2}.
    \ea
\]
Recalling the definition of \(\ft\), we complete the proof of Lemma  \ref{lem-sc-2} with \(B := 2 (\ft(\Xt(0)) - \ft^\ast)\) and \(D = 2n L^2 \sup_{\beta \geq 0} \|\Xt_{\rm nc}^\ast - \Xt^\ast_\beta\|^2\).

\section{Proof of Theorem \ref{thm-sc-dis}}
\label{app:thm2}
Defining $\xt_k:=[\xb_{1,k};\dots;\xb_{n,k}]$ and $\yt_k:=[\yb_{1,k};\dots;\yb_{n,k}]$, DNGD-SC \eqref{eq:DNGD-SC} can be written compactly as
\begin{equation}
\label{eq:DNGD-SC_com}
\ba
\xt_{k+1} &= \yt_{k} - \eta \nabla\Theta_\beta(\yt_{k}), \\
\yt_{k} &= \xt_{k} + \frac{k-1}{k+1} \left( \xt_{k} - \xt_{k-1} \right)
\ea
\end{equation}
for $\Theta_\beta(\xt_k) := \frac{1}{2} \xt_k^\top \Lc \xt_k + \beta \ft(\xt_k)$.
As Assumption \ref{ass-strcon} holds, according to the proof of Lemma  \ref{lem-sc-2} in Section \ref{sec.hfb}, we know $\Theta_\beta$ is $\mu_\theta$-strongly convex for $\mu_\theta=\min\left\{\lambda_2,\frac{\beta\mu}{n}\right\}$. In addition, as $\nabla f_i$ is $L$-Lipschitz continuous, it follows that $\nabla \tilde{f}$ is $\sqrt{n} L$-Lipschitz continuous. Then we have 
\[
\ba
&\|\nabla\Theta_\beta(\xb)-\nabla\Theta_\beta(\xb')\|\\ =& \|\Lc\xb-\Lc\xb'+\beta\ft(\xb)-\beta\ft(\xb')\|\\\leq& (\lambda_n+\beta \sqrt{n}L )\|\xb-\xb'\|
\ea
\]
for all $\xb,\xb'\in\mathbb{R}^{nd}$, which yields $\nabla\Theta_\beta$ is $\lambda_n+\beta \sqrt{n}L$-Lipschitz continuous.

As
\(\eta = \frac{1}{\lambda_n + \beta\sqrt{n} L}\), \(\alpha = 2\sqrt{\min\left\{\lambda_2, \frac{\beta \mu}{n}\right\}}\), it can be noticed that $\Theta_\beta$ is $\frac{\alpha^2}{2}$-strongly convex and $\nabla\Theta_\beta$ is $\frac{1}{\eta}$-Lipschitz continuous.
Following \cite{NEST}, it can be concluded that $\xb_{i,k}$ generated in \eqref{eq:DNGD-SC} satisfies that
\[
\Theta_\beta(\xt_{k}) - \Theta_\beta^\ast \leq 2 (\Theta_\beta(\xt_{0}) - \Theta_\beta^\ast)\left(1-\frac{\alpha^2}{2\eta}\right)^k ,
\]
where $\Theta_\beta^\ast$ is the minimal value of $\Theta_\beta$.

Next, following the same process in Section \ref{sec.hfb} to obtain
\eqref{eq:obj}, we directly obtain that
\begin{equation}
    \label{eq:obj2}
    \ba
    &\ft(\mathbf{1}_n \otimes \xb_{i,k}) - f^\ast \\
    \leq & \left( 2 (\ft(\xt_0) - \ft^\ast) + \frac{\xt_0^\top \Lc \xt_0}{\beta} \right)\left(1-\frac{\alpha^2}{2\eta}\right)^k + \beta \frac{D}{\lambda_2}
    \ea
\end{equation}
for \(D = 2n L^2 \sup_{\beta \geq 0} \|\Xt_{\rm nc}^\ast - \Xt^\ast_\beta\|^2\).

As \(\xb_{i,0}=\xb_{j,0}\) for all \(i,j\in\mathbb{V}\), we have
\(\xt_0^\top\Lc\xt_0=0\). Therefore, \eqref{eq:obj2} yields
\begin{equation}
    \label{eq:appa}
    \ba
    f(\xb_{i,k})-f^\ast
    \leq
    B\left(
    1-
    \sqrt{
    \frac{
    \min\left\{\lambda_2,\frac{\beta\mu}{n}\right\}
    }{
    \lambda_n+\beta\sqrt{n}L
    }}
    \right)^k
    +\frac{\beta D}{\lambda_2},
    \ea
\end{equation}
where \(B:=2(\ft(\xt_0)-\ft^\ast)>0\), and \(D>0\) is the
constant defined above. Both \(B\) and \(D\) are fixed once the
initial condition and the problem data are given.

Let \(\beta=c_\beta\epsilon\), where \(c_\beta>0\) is independent
of \(\epsilon\) and satisfies
\[
c_\beta\leq\frac{\lambda_2}{2D},
\qquad
c_\beta\epsilon
\leq
\min\left\{
\frac{\lambda_2}{4L^2},
\frac{\mu\lambda_2}{16n^2L^2}
\right\}.
\]
It then follows that
\[
\frac{\beta D}{\lambda_2}
=
\frac{c_\beta\epsilon D}{\lambda_2}
\leq\frac{\epsilon}{2}.
\]
Moreover, for
\[
k>
k_0(\epsilon,c_\beta)
:=
\frac{\ln\left(\frac{\epsilon}{2B}\right)}
{\ln\left(
1-
\sqrt{
\frac{
\min\left\{\lambda_2,\frac{c_\beta\epsilon\mu}{n}\right\}
}{
\lambda_n+c_\beta\epsilon\sqrt{n}L
}
}
\right)},
\]
the first term on the right-hand side of \eqref{eq:appa} satisfies
\[
B\left(
1-
\sqrt{
\frac{
\min\left\{\lambda_2,\frac{c_\beta\epsilon\mu}{n}\right\}
}{
\lambda_n+c_\beta\epsilon\sqrt{n}L
}
}
\right)^k
\leq\frac{\epsilon}{2}.
\]
Consequently, \(f(\xb_{i,k})-f^\ast\leq\epsilon\) for every
\(i\in\mathbb{V}\) and all \(k>k_0(\epsilon,c_\beta)\).

Finally, since \(c_\beta\) is a positive constant independent of
\(\epsilon\), we have \(\beta=\Theta(\epsilon)\). Thus, for
sufficiently small \(\epsilon\),
\[
k_0(\epsilon,c_\beta)
=
\mathcal{O}\left(
\sqrt{\kappa}\,
\frac{\ln\epsilon^{-1}}{\sqrt{\epsilon}}
\right),
\]
which completes the proof.
\color{black}

\section{Proof of Theorem \ref{thm-c-con}}

Flow \eqref{eq:DGDO-N-C} can be written in a compact form as

\begin{equation}
    \label{eq:DGDO-N-C_com}
    \ba
    \ddot{\Xt}+\frac{r}{t}\dot{\Xt}+\Lc\Xt+\frac{1}{t^{3-r}}\Gb(\Xt)=0,
    \ea
\end{equation}
where $r\in[2,3)$ and \(\Gb(\Xt):=[\nabla f_1(\Xb_{1});\dots;\nabla f_n(\Xb_{n})]\).
\color{black}

By Assumption \ref{ass-convex}, the optimal solution set \(\mathbb{S}\) is not empty. Assume \(\Xb^\ast \in \mathbb{S}\) is an optimal solution of Problem \eqref{eq:DO}.
We define a function \(V(\Xt(t), t)\) as
$V(\Xt(t),t)=\frac{1}{2}t^2\Xt(t)\Lc\Xt(t)+t^{r-1}(\ft(\Xt(t))-f^\ast)+\frac{1}{2}\|t\dot{\Xt}(t)+(r-1)\Xt(t)-(r-1)\Xt^\ast\|^2$, where $\Xt^\ast=\mathrm{1}_n\otimes\Xb^\ast $, then we have

\[
\ba
\dot{V}
=& (r-1)t^{r-2}(\ft(\Xt)-f^\ast)-(r-1)t^{r-2}(\Xt-\Xt^\ast)^T\Gb({\Xt)}\\&+(2-r)t\Xt\Lc\Xt.
\ea
\]
\color{black}
Since \(\ft(\tilde{\Xb}^\ast) = f^\ast\) and \(\ft\) is a convex function (being the sum of several convex functions) by Assumption \ref{ass-convex}, and $(2-r)t\Xt\Lc\Xt\leq0$, we have \(\dot{V} \leq 0\).
Therefore, we conclude that

\begin{equation}
\ba
\label{eq:boundf}
\ft(\Xt(t)) - f^\ast \leq &\ \frac{V(\Xt(t), t)}{t^{r-1}} \leq \frac{V(\Xt(0), 0)}{t^{r-1}} = \frac{\|\Xt(0) - \Xt^\ast\|^2}{2t^{r-1}},
\ea
\end{equation}
and
\begin{equation}
\label{eq:xlxbounded}
\ba
\frac{1}{2} t^2 \Xt(t)^\top \Lc \Xt(t) + t^{r-1} (\ft(\Xt(t)) - f^\ast) \leq &\ \frac{r-1}{2} \|\Xt(0) - \Xt^\ast\|^2.
\ea
\end{equation}
\color{black}
This implies that \(\lim_{t \to \infty} \Xt(t)^\top \Lc \Xt(t) = 0\), as \(\ft(\Xt(t)) - f^\ast > \ft^\ast - f^\ast\). Thus, from \eqref{eq:boundf}, we directly obtain \(\limsup_{t \to \infty} \ft(\Xt(t)) - f^\ast = 0\), which implies that \(\Xt(t)\) converges to the set \(\mathbb{S}\). 

Since \(\mathbb{S}\) is compact and non-empty by Assumption \ref{ass-convex}, we know that \(\|\Xt(t)\|\) is bounded for all \(t \geq 0\). Moreover, since \(\nabla \ft\) is Lipschitz continuous, there exists some constant \(G > 0\) such that \(\|\nabla \ft(\Xt(t))\| \leq G\) for all \(t \geq 0\).
For each \(i \in \mathbb{V}\), we have
\begin{equation}
    \label{eq:inter2}\ba
&f^\ast-\ft(\Xt(t))\\=& f^\ast-\ft(\mathbf{1}_n\otimes\Xb_i(t))+\ft(\mathbf{1}_n\otimes\Xb_i(t))-\ft(\Xt(t))\\
\leq& \nabla \ft(\Xt(t))^\top(\mathbf{1}_n\otimes\Xb_i(t))-\Xt(t))\\&+\frac{\sqrt{n}L}{2}\|\mathbf{1}_n\otimes\Xb_i(t))-\Xt(t)\|^2\\\leq&G\xi(t)+\frac{\sqrt{n}L}{2}\xi^2(t)
\leq  B  \xi(t)
\ea
\end{equation}
where \(\xi(t) := \max_{i\in\mV}\| \mathbf{1}_n \otimes \Xb_i(t) - \Xt(t) \|\) and \(B := G + \frac{\sqrt{n} L}{2} \sup_{t \geq 0} \xi(t)\), given that $\xi(t)$ is bounded as \(\|\Xt(t)\|\) is bounded for all \(t \geq 0\). The first inequality follows from \(f^\ast - \ft(\mathbf{1}_n \otimes \Xb_i(t)) \leq 0\) and from the fact that \(\nabla \tilde{f}\) is \(\sqrt{n} L\)-Lipschitz continuous. The second inequality uses \(\|\nabla \ft(\Xt(t))\| \leq G\) for all \(t \geq 0\).

{Substituting} \eqref{eq:inter2} and the fact $\Xt(t)^\top \Lc \Xt(t) \geq \frac{\lambda_2}{2} \xi_i^2(t)$ into \eqref{eq:xlxbounded}, we obtain

\[
\frac{\lambda_2}{4} (t \xi(t))^2 - B t^{r-1} \xi(t) - \frac{r-1}{2} \|\Xt(0) - \Xt^\ast\|^2 \leq 0,
\]
which leads to the bound
\begin{equation}
    \ba
    \label{eq:boundxi}
    \xi(t) \leq \frac{2B t^{r-2} + \sqrt{4B^2 t^{2r-4} + 2\lambda_2 \|\Xt(0) - \Xt^\ast\|^2}}{t \lambda_2}.
    \ea
\end{equation}

\color{black}






Finally, we can obtain that for all \(i \in \mathbb{V}\), the following holds,

\[
\ba
& \ft(\mathbf{1}_n \otimes \Xt_i(t)) - f^\ast \\
\leq & \ft(\mathbf{1}_n \otimes \Xt_i(t)) - \ft(\Xt(t)) + \ft(\Xt(t)) - f^\ast \\
\leq & \frac{1}{t} \left( \frac{2B^2 t^{r-2} + B\sqrt{4B^2 t^{2r-4} + 2\lambda_2 \|\Xt(0) - \Xt^\ast\|^2}}{\lambda_2} \right)\\&+\frac{1}{t} \left( \frac{\|\Xt(0) - \Xt^\ast\|^2}{2t^{r-2}} \right),
\ea
\]
\color{black}
where the last inequality follows from \eqref{eq:boundf}, \eqref{eq:inter2}, and \eqref{eq:boundxi} {Thus, the proof of \eqref{eq.thm2_a} is complete by recalling the definition of \(\ft\), \eqref{eq.thm2_b} can be obtained by substituting $r=2$.}

\section{Proof of Theorem \ref{thm-c-dis}}
\label{app:thm4}
DNGD-C \eqref{eq:DNGD-C} can be compactly written as 
\begin{equation}
\label{eq:DNGD-C_com}
\ba
\xt_{k+1} &= \yt_{k} - \eta \Lc \yt_{k} - \frac{\sqrt{\eta}}{k+1} \Gb(\yt_{k}), \\
\yt_{k} &= \xt_{k} + \frac{k-1}{k+1} \left( \xt_{k} - \xt_{k-1} \right).
\ea
\end{equation}

We define a Lyapunov function as
$V_{k}=V(k,\xt_k)=\frac{\eta}{2}k^2\xt_k\Lc\xt_k+\sqrt{\eta} k(\ft(\xt_k)-f^\ast)+\frac{1}{2}\|(k+1)\yt_k-k\xt_k-\Xt^\ast\|^2$.
Then, we have the following expression for \( V_{k+1} - V_k \),
\begin{equation}
\ba
\label{eq:Vdif_1}
&V_{k+1} - V_k 
\\= & \, \eta \left[ \left(k + \frac{1}{2}\right) \xt_{k+1}^\top \Lc \xt_{k+1} + \frac{k^2}{2} \left( \xt_{k+1}^\top \Lc \xt_{k+1} - \xt_k^\top \Lc \xt_k \right) \right] \\
   & + \sqrt{\eta} \left[ f(\xt_{k+1}) - f^\ast + k \left( f(\xt_{k+1}) - f(\xt_k) \right) \right] \\
   & - \left( \eta (k+1) \Lc \yt_k + \sqrt{\eta} \nabla \Gb(\yt_k) \right)^\top \left( k (\yt_k - \xt_k) + \yt_k - \Xt^\ast \right) \\
   & + \frac{\eta}{2} \left\| \sqrt{\eta} (k+1) \Lc \yt_k + \nabla \Gb(\yt_k) \right\|^2.
\ea
\end{equation}

We now introduce two important inequalities. If \( \sqrt{\eta} \leq \frac{\sqrt{L^2 + 4 \lambda_n} - L}{2 \lambda_n} \), then it holds
\begin{equation}
    \label{eq:thm2fact1}
    \ba
    &\frac{1}{k+1} f(\xt_{k+1}) + \frac{\sqrt{\eta}}{2} \xt_{k+1}^\top \Lc \xt_{k+1} 
    \\\leq & \, \frac{1}{k+1} f(\xt_k) + \left( \yt_k - \xt_k \right)^\top \left( \sqrt{\eta} \Lc \yt_k + \frac{1}{k+1} \Gb(\yt_k) \right) \\
    & - \frac{\sqrt{\eta}}{2} \left\| \sqrt{\eta} \Lc \yt_k + \frac{1}{k+1} \Gb(\yt_k) \right\|^2 + \frac{\sqrt{\eta}}{2} \xt_k^\top \Lc \xt_k,
    \ea
\end{equation}
and
\begin{equation}
    \label{eq:thm2fact2}
    \ba
    &\frac{1}{k+1} f(\xt_{k+1}) + \frac{\sqrt{\eta}}{2} \xt_{k+1}^\top \Lc \xt_{k+1} 
    \\\leq & \, \frac{1}{k+1} f^\ast + \left( \yt_k - \Xt^\ast \right)^\top \left( \sqrt{\eta} \Lc \yt_k + \frac{1}{k+1} \Gb(\yt_k) \right) \\
    & - \frac{\sqrt{\eta}}{2} \left\| \sqrt{\eta} \Lc \yt_k + \frac{1}{k+1} \Gb(\yt_k) \right\|^2 - \frac{\sqrt{\eta}}{2} \yt_k^\top \Lc \yt_k.
    \ea
\end{equation}

{\em Proof of \eqref{eq:thm2fact1} and \eqref{eq:thm2fact2}}:
We define a function \( \Phi_k(\xt) := \frac{1}{k+1} f(\xt) + \frac{\sqrt{\eta}}{2} \xt^\top \Lc \xt \). It can be observed that \( \Phi_k(\xt) \) is convex by Assumption \ref{ass-convex}, and the gradient \( \nabla \Phi_k(\xt) \) is Lipschitz continuous with Lipschitz constant \( L_\Phi := L + \sqrt{\eta} \lambda_n \). Therefore, for all \( \yt, \zt \in \mathbb{R}^{nd} \), we have the following inequalities,
\[
\ba
&\Phi_k(\yt) \leq  \, \Phi_k(\zt) + \nabla \Phi_k(\yt)(\yt - \zt), \\
&\Phi_k(\yt) \leq \, \Phi_k(\Xt^\ast) - \frac{\sqrt{\eta}}{2} \yt^\top \Lc \yt + \nabla \Phi_k(\yt)(\yt - \Xt^\ast), \\
&\Phi_k(\yt - \sqrt{\eta} \nabla \Phi_k(\yt)) \leq \, \Phi_k(\yt) + \left( \frac{\eta L_\Phi}{2} - \sqrt{\eta} \right) \| \nabla \Phi_k(\yt) \|^2.
\ea
\]
By setting \( \yt = \yt_k \) and \( \zt = \xt_k \) in the first and third inequalities, we obtain \eqref{eq:thm2fact1}. Similarly, by setting \( \yt = \yt_k \) and \( \zt = \Xt^\ast \) in the second and third inequalities, we obtain \eqref{eq:thm2fact2}.  \hfill$\square$

Substituting \( \eqref{eq:thm2fact1} \times (k^2 + k) \) and \( \eqref{eq:thm2fact2} \times (k+1) \) into \eqref{eq:Vdif_1}, we obtain the following bound for \( V_{k+1} - V_k \) as
\[
V_{k+1} - V_k \leq \eta \left( \frac{k}{2} \xt_k^\top \Lc \xt_k - \frac{k+1}{2} \yt_k^\top \Lc \yt_k \right).
\]

Next, we will prove that \( V_{k+1} - V_k \leq 0 \), which is equivalent to showing that \( (k+1)g(\yt_k) > k g(\xt_k) \) with \( g(\xt_k) := \yt_k^\top \Lc \yt_k \), a convex function.

Since \( (k+1)g(\yt_k) > k g(\xt_k) \) holds for \( k = 0 \), we now prove that \( (k+1)g(\yt_k) > k g(\xt_k) \) holds if \( k g(\yt_k) > (k-1) g(\xt_k) \) holds for \( k \geq 1 \). We consider two cases.

\begin{itemize}
    \item [(i)] If \( g(\yt_k) \geq g(\xt_k) \), the result follows immediately.
    \item [(ii)]  If \( g(\yt_k) < g(\xt_k) \), we have
   \[
   \ba
   &(k+1)g(\yt_k) - k g(\xt_k) 
   \\=& (k+1)(g(\yt_k) - g(\xt_k)) + g(\xt_k) \\
   \geq& (k+1) \nabla g(\xt_k)(\yt_k - \xt_k) + g(\xt_k) \quad  \\
   \geq& (k-1) \nabla g(\xt_k)(\xt_k - \xt_{k-1}) + g(\xt_k) \quad  \\
   \geq& k g(\xt_k) - (k-1) g(\xt_{k-1}) \\
   \geq& k g(\yt_k) - (k-1) g(\xt_{k-1}) 
   \geq 0,
   \ea
   \]
   where the first and third inequalities use the convexity of \( g \), and the second inequality follows from \eqref{eq:DNGD-C_com}. Thus, we have proven that \( (k+1)g(\yt_k) > k g(\xt_k) \).
Therefore, we conclude that \( V_{k+1} - V_k \leq 0 \), meaning that \( V_k \) is bounded.
\end{itemize}

The remaining proof of \eqref{eq:concl} is similar to that of Theorem \ref{thm-c-con}, so it is omitted here.

Finally, we analyze the case where \( \eta = \left( \frac{\sqrt{L^2 + 4\lambda_n} - L}{2\lambda_n} \right)^2 \). We have
\[
\ba
&f(\xb_{i,k}) - f^\ast\\\leq& \frac{2\kappa B}{k\sqrt{L^2+4\lambda_2\kappa}-L}\big(2B+\sqrt{4B^2+2\lambda_2\sum_{i\in\mV}\|\xb_{i,0}-\Xb^\ast\|^2}\\&+\lambda_2\sum_{i\in\mV}\|\xb_{i,0}-\Xb^\ast\|^2\big).
\ea
\]
Thus, we conclude that for any fixed \( \lambda_2 \), it holds
\[
f(\xb_{i,k}) - f^\ast = \mathcal{O} \left( \frac{\sqrt{\kappa}}{k} \right).
\]


\begin{thebibliography}{99}





\bibitem{magnusbook} M. Mesbahi and M. Egerstedt. {\em Graph Theoretic Methods in Multiagent Networks}. Princeton University Press, 2010.

 \bibitem{martinez07} S. Martinez, J. Cort\'{e}s, and F. Bullo, ``Motion Coordination with Distributed
Information,'' {\em IEEE Control Systems Magazine}, vol. 27, no. 4, pp.
75-88, 2007.

\bibitem{kar12} S. Kar, J. M. F. Moura and K. Ramanan, ``Distributed Parameter  Estimation in Sensor Networks:
Nonlinear Observation Models and Imperfect Communication,''  {\em IEEE Transactions on Information Theory}, vol. 58, no. 6, pp. 3575-52, 2012.


\bibitem{Rabbat2010} A. G. Dimakis,  S. Kar,
J. M. F. Moura,  M. G. Rabbat,  and
A Scaglione,  ``Gossip Algorithms for
Distributed Signal Processing,'' {\em Proceedings of IEEE}, vol. 98, no. 11, pp. 1847-1864, 2010.


\bibitem{Patari2022}
N. Patari, V. Venkataramanan, A. Srivastava, D. Molzahn, N. Li, A. Annaswamy, ``Distributed Optimization in Distribution Systems: Use Cases, Limitations, and Research Needs,'' {\em IEEE Transactions on Power Systems}, vol. 37, no. 5, pp. 3469–3481, 2022. 

\bibitem{Koloskova2021}
A. Koloskova, T. Lin, S. U. Stich, ``An Improved Analysis of Gradient Tracking for Decentralized Machine Learning,'' in {\em Proceedings of Neural Information Processing Systems}, pp. 11422–11435, 2021. 

\bibitem{Shen2022}
J. Shen, E. K. H. Kammara, L. Du, ``Nonconvex, Fully Distributed Optimization Based CAV Platooning Control under Nonlinear Vehicle Dynamics,'' {\em IEEE Transactions on Intelligent Transportation Systems}, 2022. 

\bibitem{FDGM} D. Jakovetić, J. Xavier and J. M. F. Moura, ``Fast Distributed Gradient Methods,'' {\em IEEE Transactions on Automatic Control}, vol. 59, no. 5, pp. 1131-1146, 2014.

\bibitem{DNGA}
D. Jakovetić, J. Moura and J. Xavier, ``Distributed Nesterov-like Gradient Algorithms,'' in {\em Proceedings of IEEE Conference on Decision and Control}, pp. 5459-5464, 2012.



\bibitem{chen2025distributed}
Y. Chen, A. L. Fradkov, K. Fu, X. Fu, and T. Li,
``Distributed Stochastic Optimization With Unbounded Subgradients Over Randomly Time-Varying Networks,''
\textit{IEEE Transactions on Automatic Control}, vol. 70, no. 6, pp. 4008--4015, Jun. 2025, doi: 10.1109/TAC.2024.3525182.

\bibitem{FDL}
Brendan McMahan, Eider Moore, Daniel Ramage, Seth Hampson, and Blaise Aguera y Arcas,
``Communication-Efficient Learning of Deep Networks from Decentralized Data,''
in \emph{Proceedings of the 20th International Conference on Artificial Intelligence and Statistics (AISTATS)},
PMLR, volume 54, pp. 1273--1282, 2017.

\bibitem{DSMF}
A. Nedic and A. Ozdaglar, ``Distributed Subgradient Methods for Multi-agent Optimization,'' {\em IEEE Transaction on Automatic Control}, vol. 54, no. 1, pp. 48–61, 2009.


\bibitem{MB-GT}
M. Bin, I. Notarnicola, L. Marconi and G. Notarstefano, ``A System Theoretical Perspective to Gradient-tracking Algorithms for Distributed Quadratic Optimization,'' In {\em Proceedings of IEEE Conference on Decision and Control}, pp. 2994-2999, 2019.


\bibitem{AADS}
H. Hendrikx, F, Bach and L. Massoulié, ``An Accelerated Decentralized Stochastic Proximal
Algorithm for Finite Sums,'' in {\em Proceedings of Neural Information Processing Systems}, vol. 32, pp. 952–962, 2019.



\bibitem{Xu2020}
J. Xu, Y. Tian, Y. Sun, and G. Scutari, ``Accelerated Primal-Dual Algorithms for Distributed Smooth Convex Optimization over Networks,'' in {\em Proceedings of the International Conference on Machine Learning}, vol. 108, pp. 10562–10572, 2020.


\bibitem{ADNG}
G. Qu and N. Li, ``Accelerated Distributed Nesterov Gradient Descent,'' {\em IEEE Transactions on Automatic Control}, vol. 65, no. 6, pp. 2566-2581, 2020.


\bibitem{AGAA}
R. Xin and U. A. Khan, ``Distributed Heavy-Ball: A Generalization and Acceleration of First-Order Methods With Gradient Tracking,'' {\em IEEE Transactions on Automatic Control}, vol. 65, no. 6, pp. 2627-2633, 2020.


\bibitem{ADEF}
W. Su, S. Boyd and E. Candes, ``A Differential Equation for Modeling Nesterov's Accelerated Gradient Method: Theory and Insights,''
{\em Journal of Machine Learning Research}, vol.17, pp. 1-43, 2016.

\bibitem{ROCO}
H. Attouch, Z. Chbani, and H. Riahi, ``Rate of Convergence of the Nesterov
Accelerated Gradient Method in the Subcritical Case $\alpha\leq 3$,'' {\em ESAIM: Control, Optimization Calculus Variations}, vol. 25, pp. 1–34, 2019.

\bibitem{large}
H. Attouch, Z. Chbani, J. Peypouquet, and P. Redont, ``Fast Convergence
 of Inertial Dynamics and Algorithms with Asymptotic Vanishing Viscosity,''
 {\em Mathematical Programming}, vol. 168, no. 1/2, pp. 123–175, 2018.
 














\bibitem{DPDN}
X. Zeng, J. Lei and J. Chen, ``Dynamical Primal-Dual Nesterov Accelerated Method and Its Application to Network Optimization,'' {\em IEEE Transactions on Automatic Control}, vol. 68, no. 3, pp. 1760-1767, 2023.


\bibitem{liang2018}
S. Liang, X. Zeng, and Y. Hong, ``Distributed Nonsmooth Optimization with Coupled Inequality Constraints via Modified Lagrangian Function,'' \emph{IEEE Transactions on Automatic Control}, vol. 63, no. 6, pp. 1753--1759, 2018.

\bibitem{OAFS}
J. Wang and N. Elia, ``A Control Approach to Distributed Optimization,'' in \emph{2010 48th Annual Allerton Conference on Communication, Control, and Computing}, pp. 557-561. 2010,

\bibitem{cort}
B. Gharesifard and J. Cortés, ``Distributed Continuous-time Convex Optimization on Weight-balanced Digraphs,'' {\em IEEE Transactions on Automatic Control}, vol. 59, no. 3, pp. 781-786, 2014.

\bibitem{Yang}
S. Yang, Q. Liu and J. Wang, ``A Multi-agent System with A Proportional-integral Protocol for Distributed Constrained Optimization,'' {\em IEEE Transactions on Automatic Control}, vol. 62, no. 7, pp. 3461-3467, 2017.


\bibitem{Shao}
J. Liu, S. Mou and A. S. Morse, ``An Asynchronous Distributed Algorithm for Solving a Linear Algebraic Equation,'' In {\em Proceedings of IEEE Conference on Decision and Control}, pp. 5409-5414, 2013.

\bibitem{Shi}
G. Shi, B. D. O. Anderson, and U. Helmke, ``Network Flows That Solve Linear Equations,'' \emph{IEEE Transactions on Automatic Control}, vol. 62, no. 6, pp. 2659–2674, 2017.

\bibitem{Gadjov}
D. Gadjov and L. Pavel, ``A Passivity-Based Approach to Nash Equilibrium Seeking Over Networks,'' \emph{IEEE Transactions on Automatic Control}, vol. 64, no. 3, pp. 1077–1092, 2019.


\bibitem{Ye}
M. Ye and G. Hu, ``Distributed Nash Equilibrium Seeking by a Consensus-Based Approach,'' \emph{IEEE Transactions on Automatic Control}, vol. 62, no. 9, pp. 4811–4818, 2017.

\bibitem{LW_SUR}
L. Wang and G. Shi, ``Distributed Flow Computing,'' {\em Foundations and Trends in Systems and Control}, vol. 12, no 1, 2025.

\bibitem{NEST}
Y. Nesterov. {\em Introductory Lectures on Convex Optimization: A Basic Course}, Kluwer Academic Publishers, Boston, 2004.

\bibitem{AVPO}
A. Wibisono, A. Wilson, M. Jordan, ``A Variational Perspective on Accelerated Methods in Optimization,'' {\em Proceedings of the National Academy of Sciences of the United States of America}, vol.113, no. 47, pp. E7351-E7358, 2016.

\bibitem{EXTRA}
W. Shi, Q. Ling, G. Wu, and W. Yin, ``EXTRA: An Exact First-order Algorithm for Decentralized Consensus Optimization,'' {\em SIAM Journal on Optimization}, vol. 25, no. 2, pp. 944-966, 2015.

\bibitem{XY-LCOF}
X. Yi, S. Zhang, T. Yang, T. Chai and K. H. Johansson, ``Linear Convergence of First- and Zeroth-order Primal–dual Algorithms for Distributed Nonconvex Optimization,'' {\em IEEE Transactions on Automatic Control}, vol. 67, no. 8, pp. 4194-4201, 2022.




\bibitem{Standard_function}
G. Wu, R. Mallipeddi and P. Suganthan, ``Problem Definitions and Evaluation Criteria for the CEC 2017 Competition on Constrained Real-Parameter Optimization''. National University of Defense Technology, Kyungpook National University and Nanyang Technological, 2017.

\bibitem{Ren2026} Z. Ren et al., ``Distributed Optimization by Network Flows With Spatio-Temporal Compression,'' in \textit{IEEE Transactions on Automatic Control}, doi: 10.1109/TAC.2026.3654479.

\end{thebibliography}
\end{document}